\documentclass[superscriptaddress]{revtex4-2}
\usepackage{braket}
\usepackage{graphicx}
\usepackage{eufrak}
\usepackage{amsmath,mathtools}
\usepackage{dsfont}
\usepackage[font={small}]{caption}
\usepackage{subcaption}
\usepackage{amssymb,epsfig,color,xspace,enumitem,listings}
\graphicspath{ {./Image/} }
\usepackage{overpic}
\usepackage{physics}
\usepackage{amsthm}
\usepackage{thmtools}
\usepackage{thm-restate}
\usepackage{pstricks}
\usepackage{float}
\usepackage{xcolor}

\usepackage{microtype} % Improves TeX's choices of linebreaks, text justification, end-of-line hyphenation, kerning etc.

\usepackage{hyperref} % Package compatibility note: hyperref should be included after most other packages unless otherwise specified.
\usepackage[capitalize]{cleveref} % To be included after hyperref.
\usepackage{xparse}

\hypersetup{
    colorlinks = true,
    linkcolor = blue,
    citecolor = blue,
    urlcolor = blue
}

\renewcommand{\H}{\mathcal{H}}
\newcommand{\K}{\mathcal{K}}
\newcommand{\B}{\mathcal{B}} % Corresponds to the set of bounded operators
\newcommand{\Pos}{\textrm{Pos}} % Set of Positive operators
\newcommand{\Hin}{\H_{\mathrm{in}}}      % input Hilbert space
\newcommand{\Kout}{\K_{\mathrm{out}}}      % output Hilbert space
\newcommand{\Phikeep}{\Phi_{\mathrm{keep}}} % Channel going only to the kept detectors

\newtheorem{lemma}{Lemma}[section]

\newtheorem{theorem}[lemma]{Theorem}
\newtheorem{corollary}[lemma]{Corollary}
\newtheorem{remark}{Remark}
\theoremstyle{definition}
\newtheorem{definition}[lemma]{Definition}
\theoremstyle{plain}

\begin{document}

\title{Enforcing IID structure on time-bin encoded QKD protocols via coarse-graining}

\author{Shlok Nahar}
\thanks{These authors contributed equally to this work.}
\affiliation{Institute for Quantum Computing and Department of Physics and Astronomy, University of Waterloo}

\author{Shihong Pan}
\thanks{These authors contributed equally to this work.}
\affiliation{
    Institute for Quantum Computing and Cheriton School of Computer Science, University of Waterloo
}
\altaffiliation{
    Current address: New York University Shanghai, 567 West Yangsi Road, Shanghai, 200126, China
}

\author{Norbert L\"utkenhaus}
\affiliation{Institute for Quantum Computing and Department of Physics and Astronomy, University of Waterloo}

\date{\today}

\begin{abstract}
    Many security proofs for quantum key distribution (QKD) require Bob's measurement to have a tensor-product structure across protocol rounds, with some techniques requiring the stronger independent-and-identically-distributed (IID) condition. Time-bin encoded protocols often rely on interferometers whose detector outcomes depend on the interference between optical modes from neighbouring rounds, obstructing the direct application of such proofs. We show that classical post-processing of Bob's measurement data --- specifically, discarding the outcomes of detectors sensitive to inter-round coherence --- is sufficient to recover a product measurement positive operator-valued measure (POVM) (which is IID when the same single-round setup is used in every round). Applied to the Mach-Zehnder interferometer and the IID variant of the COW detection setup, this removes the need for the additional vacuum pulse introduced in prior analyses to establish tensor product structure of the measurement POVM, recovering better key rates without placing any restriction on Eve's attack.
\end{abstract}

\maketitle

\section{Introduction}
\label{sec:intro}

Quantum key distribution (QKD) is a method for realising quantum-safe cryptography \cite{mosca2018cybersecurity}. QKD security proofs \cite{koashi_simple_2005,koashi_simple_2009,tomamichel_uncertainty_2011,christandl_post-selection_2009,nahar_postselection_2024,arqand_marginal-constrained_2025,dupuis_entropy_2020,metger_security_2023, tupkary2025qkd} promise future-proof security by bounding the information accessible to an adversary. These proofs, however, require the protocol implementation to satisfy a number of structural properties. In particular, they typically assume that the measurements made in the protocol are independent across rounds \footnote{Throughout this paper, ``product'' or ``independent'' means that the multi-round POVM elements factor as tensor products over rounds. We reserve ``IID'' for the stronger case in which the \emph{same} single-round POVM is used in every round. The postselection technique \cite{christandl_post-selection_2009,nahar_postselection_2024} requires this stronger identical-round assumption.}, and only limited proofs \cite{sandfuchs2025security} exist for protocols \cite{inoue_differential_2002,stucki_fast_2005} that do not satisfy this structure.

Time-bin encoded QKD protocols often use interferometers that measure coherence between adjacent time-bins, which introduces correlations across protocol rounds. This violates the round-wise product-measurement assumption required by many standard security proofs. In this work, we prove that mathematically coarse-graining the measurement data --- specifically over the detector events sensitive to inter-round coherence --- renders the resulting effective POVM a tensor product over rounds. In effect, the coarse-graining removes the measurement outcomes that couple neighbouring rounds while retaining those needed to characterise each round individually. When the same single-round setup is used in every round, this tensor product is IID. As a result, standard security analyses can be applied to these time-bin encoded protocols without hardware modifications or additional vacuum pulses.

Prior works \cite{lavie_improved_2022, sulimany_highcow_2025} address the non-independence of the measurement by sending an additional vacuum pulse after every protocol round, which keeps the protocol in the standard round-wise setting. This countermeasure reduces the key produced per unit time by a factor of $2/3$ \footnote{This fraction is $d/(d+1)$, where $d$ is the number of pulses sent per protocol round. The improvement therefore shrinks as the number of pulses per round grows.}. Moreover, the argument using the additional vacuum pulse assumes that Eve forwards the vacuum pulse to Bob. This is a restriction on Eve's attack. Thus, our work improves on prior work on this point, both in performance and in rigour.

The remainder of the paper is organised as follows. In \Cref{subsec:notation}, we introduce notation for threshold-detection setups indexed by bit strings and define coarse-graining in \cref{subsec:coarse-graining}. In \cref{subsec:structural-conditions}, we then view the entire detection setup in two steps: the detector channel from Bob's input space to the detector-output space, and the detector POVM acting on this output space. This breakup lets us state simple conditions individually on the detector channel and output POVM, respectively, that give the resultant coarse-grained product POVM (\cref{thm:generic-product-povm}). In \cref{subsec:verifying-structural-conditions}, we describe how these conditions can be readily verified for practical linear optical detection setups. Finally, in \Cref{sec:Applications}, we apply these results to the Mach-Zehnder interferometer and to the three-state detection setup, and briefly discuss implications for the QKD key rate analysis.

\section{Formalism}
\label{sec:Formalism}

We proceed in increasing levels of concreteness. We begin with a general detection setup, introduce coarse-graining, and identify two structural conditions that together imply a product POVM. We then show how these conditions are readily verified in commonly used optical detection setups.

\subsection{Notation}
\label{subsec:notation}

We denote the set of bounded linear operators acting on Hilbert space $\mathcal{X}$ as $\B(\mathcal{X})$ and the set of positive semidefinite operators as $\Pos(\mathcal{X})$. We denote the set of the first $l$ natural numbers $[l] \coloneqq \{1, \ldots, l\}$, where $l$ is a positive integer.

We consider detection setups with $t$ distinct detectors with binary \footnote{Strictly speaking, our results do not require the detectors to have binary outcomes; it suffices that each detector have a finite number of outcomes. In that more general setting, one would work with dit strings rather than bit strings. For pedagogical clarity, however, we restrict our attention to binary-outcome detectors and therefore use bit strings throughout.} (click/no-click) outcomes. A measurement outcome on the $t$ detectors is described by a bit string $\vec{x} \in \{0,1\}^t$: the $k^\text{th}$ bit $x_k = 1$ indicates a click in detector $k$, and $x_k = 0$ indicates no click. Note that in this work, we often use the term detector to refer to a detector instance in the measurement setup. Distinct detector indices need not correspond to distinct physical detector devices; for example, repeated uses of the same physical detector at different times are treated as distinct detectors.

We will describe these detection setups in terms of input and output Hilbert spaces, analogous to the input and output modes of an optical detection setup. We further consider the input Hilbert space to be partitioned into $n$ rounds:
\begin{equation}
    \Hin = \bigotimes_{j=1}^n \H_j,
\end{equation}
where $\H_j$ is the Hilbert space corresponding to round $j$. The output Hilbert space $\Kout$ can be described in terms of spaces $\K_i$ corresponding to the $i^\text{th}$ detector:
\begin{equation}
    \Kout = \bigotimes_{i=1}^t \K_i.
\end{equation}
In general, given a set $X\subset [t]$, we will describe the output Hilbert space corresponding to the detectors in the set $X$ as $\K_X \coloneqq \bigotimes_{i\in X} \K_i$.

Thus, we will describe the detection setup POVM $\left\{\Gamma_{\vec{x}}\right\}_{\vec{x}\in \{0,1\}^t}\subset \Pos(\Hin)$ in terms of two objects:
\begin{enumerate}
    \item The detector channel $\Phi$ from the input Hilbert space to the output Hilbert space:
    \begin{equation}
        \Phi : \B(\Hin) \xrightarrow[]{} \B(\Kout).
    \end{equation}
    
    \item The POVM $\left\{N_{\vec{x}}\right\}_{\vec{x}\in \{0,1\}^t}\subset \Pos(\Kout)$ acting on the output space.
\end{enumerate}
Thus, we have
\begin{equation} \label{eq:fineInputOutputChannelPOVM}
    \Gamma_{\vec{x}} = \Phi^\dag(N_{\vec{x}}) \quad \forall\ \vec{x}\in\{0,1\}^t.
\end{equation}

In general, we will use Greek letters ($\Gamma$, $\chi$, $\mu$) to denote POVM elements acting on the input Hilbert space and Latin letters ($M$, $N$, $S$) to denote POVM elements acting on the output Hilbert space.

% For coherent-state inputs we adopt the following conventions. We write a coherent state on several modes compactly as $\ketbra{\vec{\alpha}} \coloneqq \ketbra{\alpha_1} \otimes \ketbra{\alpha_2} \otimes \cdots$, where $\vec\alpha$ collects the amplitudes of those modes. An $n$-round product coherent state is written $\bigotimes_{j=1}^n \ketbra{\vec{\alpha}_j}$, where $\vec\alpha_j$ is the tuple of amplitudes of the input modes in round $j$.

\subsection{Coarse-graining}
\label{subsec:coarse-graining}

\textbf{Coarse-graining} over a subset $D \subset [t]$ is in our case defined as the process of discarding (ignoring) the measurement outcomes of the detectors in $D$.

As a concrete example, consider detectors indexed by $[3] = \{1,2,3\}$. Coarse-graining over detector $2$, i.e.\ taking $D = \{2\}$, produces a coarse-grained bit string $\vec{y} = (y_1, y_2)$, where $y_1$ corresponds to detector $1$ and $y_2$ to detector $3$. The corresponding coarse-grained POVM element is
\begin{equation}
    \chi_{\vec{y}}
    =
    \Gamma_{(y_1,0,y_2)} + \Gamma_{(y_1,1,y_2)}.
\end{equation}

In general, let $K = \{r_1, r_2, \ldots, r_k\} \subseteq [t]$ be the sorted set of \emph{kept} detector indices, where $k = \abs{K} =  t - \abs{D}$. For a coarse-grained outcome string $\vec{y} \in \{0,1\}^k$, the POVM element $\chi_{\vec{y}}$ is the sum over all fine-grained strings $\vec{x} \in \{0,1\}^t$ that agree with $\vec{y}$ on the kept detectors:
\begin{equation}
\label{eq:chiyDef}
    \chi_{\vec{y}}
    =
    \sum_{\substack{
        \vec{x} \in \{0,1\}^t \\
        x_{r_j} = y_j\ \forall\, j \in [k]
    }}
    \Gamma_{\vec{x}}.
\end{equation}

In this paper, we will assign each of the kept detectors to a round \footnote{If the number of detectors are less than the number of rounds, we would formally require some of these sets $K_j$ to be empty sets. However, for any practically relevant case we would expect that the number of detectors is at least as large as the number of rounds and so we do not comment on this edge case   beyond this footnote.}: $K = \bigcup_{j=1}^n K_j$, where each $K_j$ is the set of detector indices corresponding to round $j$ \footnote{As each detector is assigned to a single round, the sets $K_j$ are pairwise disjoint.}.
We then write the retained outcome as $\vec{y} = (\vec{y}_1, \ldots, \vec{y}_n)$, where $\vec{y}_j$ is the click pattern of the retained detectors assigned to round $j$.

Following our convention for naming POVMs on input and output spaces, we denote the coarse-grained POVMs on the input and output spaces:
\begin{equation}
    \begin{aligned}
        \Gamma_{\vec{x}} &\xrightarrow[\textrm{over }  D]{\textrm{coarse-grain}} \chi_{\vec{y}} \in \Pos(\Hin)\\
        N_{\vec{x}} &\xrightarrow[\textrm{over }D]{\textrm{coarse-grain}} M_{\vec{y}} \in \Pos(\Kout).
    \end{aligned}
\end{equation}
Thus, analogously to \cref{eq:fineInputOutputChannelPOVM}, we have for the coarse-grained POVM elements the relation
\begin{equation} \label{eq:coarseInputOutputChannelPOVM}
    \chi_{\vec{y}} = \Phi^\dag(M_{\vec{y}}) \quad \forall\ \vec{y}\in\{0,1\}^k.
\end{equation}

Finally, we define the detector channel corresponding only to the kept detectors as 
\begin{equation} \label{eq:defOfPhiKeep}
    \Phikeep \coloneqq \Tr_{\K_i:\,i\in D} \circ\ \Phi.
\end{equation}
As discussed below in \cref{def:round-local-retained-channel}, this will be used to define one of the two structural conditions required to show that the coarse-grained measurement POVM has tensor product structure.

\subsection{Structural conditions}
\label{subsec:structural-conditions}

We now isolate the two structural conditions needed to obtain a product POVM after coarse-graining. The first is a statement about the detector channel $\Phi$. The second is a statement about the coarse-grained POVM $\left\{N_{\vec{x}}\right\}_{\vec{x}\in \{0,1\}^t}\subset \Pos(\Kout)$ acting on the output space.

\begin{definition}[Round-local detector channel]
\label{def:round-local-retained-channel}
We say that the detector channel corresponding to the kept detectors is \emph{round-local} if there exist single-round channels from the input space corresponding to round $j$ to the output space corresponding to the detectors assigned to round $j$,
\begin{equation}
    \Phi_j:\mathcal \B(\H_j)\to \mathcal \B(\mathcal K_{K_j})
\end{equation}
such that
\begin{equation}
\label{eq:round-local-channel}
    \Phi_{\mathrm{keep}}
    =
    \bigotimes_{j=1}^n \Phi_j .
\end{equation}
Here, $K_j$ is the set of kept detector indices corresponding to round $j$.
\end{definition}

This is a condition only on the \emph{kept} detector-output modes; the full detector channel need not be round-local. For example, for an optical detection setup the discarded detector modes may depend on optical modes from neighbouring rounds.

\begin{definition}[Round-memoryless detectors]
\label{def:ncrc}
We say that the output POVM is round-memoryless if
\begin{equation}
\label{eq:no-cross-round-detector-correlations}
    M_{\vec{y}}
    =
    \left(
        \bigotimes_{j=1}^n
        S_{\vec{y}_j}^{(j)}
    \right)
    \otimes
    \mathbb{I},
\end{equation}
where $\{M_{\vec{y}}\}_{\vec{y}}$ is the coarse-grained POVM on the entire output space $\Kout$, $\{S_{\vec{y}_j}^{(j)}\}_{\vec{y}_j}\subset \Pos(\K_{K_j})$ is a single-round POVM on the output space, and the identity acts on $\K_{D}$.
\end{definition}

Intuitively, round-memoryless detectors correspond to detectors which have no correlations \emph{across} rounds.
However, these round-memoryless detectors can have arbitrary detector correlations \emph{within} a single round. Such detectors may have, for example, afterpulsing or dead-time effects between detectors in the same round, but not across rounds.
\begin{remark}
    Round-memoryless detectors might be realised in physical setups by having a larger time delay between rounds than within rounds. Alternately, for the purposes of QKD security proofs, this can also be accomplished by discarding rounds following detection events \cite[Chapter 7]{nahar_phd_2025}\cite{tupkary2025phase,wang2025phase}.
\end{remark}

The product structure of the coarse-grained POVM follows immediately from these two conditions, as we now show.
\begin{theorem}
\label{thm:generic-product-povm}
    Suppose that the detector channel $\Phikeep$ corresponding to the kept detectors is round-local (\Cref{def:round-local-retained-channel}), and that the coarse-grained POVM is round-memoryless (\Cref{def:ncrc}). Let $\{\chi_{\vec y}\}_{\vec y\in\{0,1\}^t}$ be the POVM (acting on the input space) obtained by coarse-graining over the discarded detector outcomes.
    Then the coarse-grained POVM factorises across rounds:
    \begin{equation}
    \label{eq:generic-product-povm}
        \chi_{\vec y}
        =
        \bigotimes_{j=1}^n
        \mu_{\vec y_j}^{(j)},
    \end{equation}
    where $\{\mu_{\vec y_j}^{(j)}\}_{\vec y_j}$ is a single-round POVM on $\H_j$. If the same single-round detector channel and the same single-round POVM (on the output space) are used in every round, then the product POVM is also IID.
\end{theorem}

\begin{proof}

The result follows from the chain of equalities below:
\begin{align}
    \label{eq:input-outputChannelDefUsed} \chi_{\vec y} &= \Phi^\dag(M_{\vec{y}})\\
    \label{eq:roundmemorylessUsed} &= \Phi^\dag\left(
        \bigotimes_{j=1}^n
        S_{\vec{y}_j}^{(j)}
    \otimes
    \mathbb{I}\right)\\
    \label{eq:defOfKeptChannelUsed} &= \Phikeep^\dag\left(
        \bigotimes_{j=1}^n
        S_{\vec{y}_j}^{(j)}\right)\\
    \label{eq:roundlocalUsed} &= \bigotimes_{j=1}^n
        \Phi_j^\dag\left(S_{\vec{y}_j}^{(j)}\right)\\
    &\eqqcolon \bigotimes_{j=1}^n
        \mu_{\vec y_j}^{(j)}.
\end{align}
Here, \cref{eq:input-outputChannelDefUsed} follows from \cref{eq:coarseInputOutputChannelPOVM}, \cref{eq:roundmemorylessUsed} follows from the assumption that the coarse-grained POVM is round memoryless (defined in \cref{eq:no-cross-round-detector-correlations}), \cref{eq:defOfKeptChannelUsed} follows from \cref{eq:defOfPhiKeep} by taking adjoints and using the fact that the adjoint of the partial trace is given by tensoring with the identity, and \cref{eq:roundlocalUsed} follows from the assumption that $\Phikeep$ is round-local (defined in \cref{eq:round-local-channel}).
This gives us the product structure.

It is clear from the above chain of equalities that if the same $\Phi_j$ and the same single-round output POVM $S_{\vec{y}_j}$ are used in every round, then single-round input POVM $\mu_{\vec{y}_j}$ is identical for every $j$, giving IID structure.

\end{proof}

\subsection{Verifying the structural conditions in optical setups}
\label{subsec:verifying-structural-conditions}

While the conditions in \Cref{thm:generic-product-povm} may appear non-trivial at first sight, they are typically straightforward to verify in concrete optical detection setups. We now discuss how these conditions can be verified for some practically-relevant classes of optical detection setups.

\subsubsection{Passive linear optical detection setups}

We first consider passive linear optical detection setups with threshold detectors. Given a set $K$ of detector indices that will be kept, we assign each of the kept detectors to a round $j$: $K = \bigcup_{j=1}^n K_j$.

The round-memoryless detector condition is equivalent to the statement that the threshold detectors have no correlations, such as dead time or afterpulsing, across rounds. Note that there can be dead time and afterpulsing effects within the same round. This can be readily verified by combining a suitable characterisation of the detectors with an appropriate choice of the time gap between rounds.

All that remains is to verify that the detector channel $\Phikeep$ corresponding to the kept detectors is round-local. We establish this using two well-known properties of (lossy) linear optical setups: they map coherent states to coherent states, and coherent states form an overcomplete basis for the space of operators. The evolution of coherent states through a linear optical setup, or equivalently the associated mode transformation, is typically straightforward to determine. If, for every detector $r\in K_j$ that is kept and assigned to round $j$, the corresponding output mode depends only on the input modes of round $j$, then the induced channel factorises across rounds. Indeed, the action of the channel on coherent states is already product, and the conclusion extends to arbitrary operators because coherent states form an overcomplete basis for the operator space.

We summarise this in the following formal statement.
\begin{corollary}[Passive linear optical detection setups] \label{cor:passive-linear-optics}
    Consider a (lossy) passive linear optical detection setup with threshold detectors. Let $K=\bigcup_{j=1}^n K_j$ be the set of kept detectors, where $K_j$ denotes the detectors assigned to round $j$.
    Suppose that:
    \begin{enumerate}
        \item the threshold detectors exhibit no correlations across rounds (e.g. due to dead time or afterpulsing); and
        \item for every detector $r\in K_j$, the corresponding output mode does not depend on the input modes of any other round $i\neq j$.
    \end{enumerate}
    Then the coarse-grained POVM obtained by discarding the outcomes of the detectors not in $K$ factorises across rounds:
    \begin{equation}
        \chi_{\vec y} = \bigotimes_{j=1}^n \mu_{\vec y_j}^{(j)},
    \end{equation}
    where $\{\mu_{\vec y_j}^{(j)}\}_{\vec y_j}$ is a single-round POVM on $\H_j$.
    If, in addition, identical mode transformation and threshold detectors are used in every round, then the resulting POVM is IID.
\end{corollary}

\subsubsection{Active basis choice}

We now extend the above to settings where the linear optical mode transformation in each round depends on a classical random variable, as is the case for detection setups where the measurement basis is chosen based on a classical random variable instead of a passive beam splitter.

Let $K = \bigcup_{j=1}^n K_j$ be the set of kept detectors, and suppose that in round $j$ an independent random variable $\Theta_j$ is drawn from some distribution. We require that, for every detector $r \in K_j$, the corresponding output mode depends neither on the input modes of any other round $i \neq j$, nor on $\Theta_i$ for $i \neq j$. Conditioned on any realisation $\vec{\theta} = (\theta_1, \ldots, \theta_n)$, the conditions of Corollary~\ref{cor:passive-linear-optics} are therefore satisfied, and the POVM factorises as
\begin{equation}
    \chi_{\vec{y}} = \sum_{\vec{\theta}} p(\vec{\theta}) \bigotimes_{j=1}^n \mu_{\vec{y}_j}^{(j)}(\theta_j),
\end{equation}
where $\{\mu_{\vec{y}_j}^{(j)}(\theta_j)\}_{\vec{y}_j}$ is the single-round POVM on $\mathcal{H}_j$ conditional on $\Theta_j = \theta_j$. Since the $\Theta_j$ are independent, the sum over $\vec{\theta}$ factorises, giving
\begin{equation}
    \chi_{\vec{y}} = \bigotimes_{j=1}^n \bar{\mu}_{\vec{y}_j}^{(j)},
\end{equation}
where $\bar{\mu}_{\vec{y}_j}^{(j)} := \sum_{\theta_j}\!p(\theta_j)\mu_{\vec{y}_j}^{(j)}(\theta_j)$ is the $\Theta_j$-averaged single-round POVM element on $\mathcal{H}_j$.

We summarise this in the following formal statement.

\begin{corollary}[Active basis choice] \label{cor:active-basis-choice}
    Consider a (lossy) linear optical detection setup with threshold detectors, where the mode transformation is determined by a sequence of classical random variables $\{\Theta_j\}_{j=1}^n$. Let $K=\bigcup_{j=1}^n K_j$ be the set of kept detectors, where $K_j$ denotes the detectors assigned to round $j$.
    Suppose that:
    \begin{enumerate}
        \item the random variables $\Theta_1, \ldots, \Theta_n$ are mutually independent;
        \item the threshold detectors exhibit no correlations across rounds (e.g.\ due to dead time or afterpulsing); and
        \item for every detector $r\in K_j$, the corresponding output mode does not depend on the input modes of any other round $i\neq j$, nor on $\Theta_i$ for $i \neq j$.
    \end{enumerate}
    Then the coarse-grained POVM obtained by discarding the outcomes of the detectors not in $K$ factorises across rounds:
    \begin{equation}
        \chi_{\vec{y}} = \bigotimes_{j=1}^n \bar{\mu}_{\vec{y}_j}^{(j)},
    \end{equation}
    where $\{\mu_{\vec{y}_j}^{(j)}(\theta_j)\}_{\vec{y}_j}$ is the single-round POVM conditional on $\Theta_j = \theta_j$ guaranteed by Corollary~\ref{cor:passive-linear-optics}, and $\bar{\mu}_{\vec{y}_j}^{(j)} := \sum_{\theta_j} p(\theta_j)\, \mu_{\vec{y}_j}^{(j)}(\theta_j)$ is the $\Theta_j$-averaged single-round POVM on $\mathcal{H}_j$.
    If, in addition, the distribution of $\Theta_j$ and the conditional mode transformation and threshold detectors are identical in every round, then the resulting POVM is IID.
\end{corollary}

\section{Applications}
\label{sec:Applications}

We now make concrete the application of our results by applying them to common optical detection setups. For the purpose of these examples, we assume that the threshold detectors have no dead times or afterpulsing. 

\subsection{The Mach-Zehnder interferometer}
\label{subsec:MZI}

Since the standard Mach–Zehnder interferometer (MZI) is a simple but important building block in many time-bin encoded QKD protocols, we first analyse its action on a train of time-bin pulses. In a MZI with a delay equal to the time-bin separation, interference occurs between adjacent time-bins. Denoting the time-bins in round $j$ by the early mode $e_j$ and the late mode $l_j$, the MZI creates interference terms such as $l_j+e_{j+1}$, which straddle the round boundary.

\begin{figure}[!ht]
    \centering
    \includegraphics[width=\linewidth]{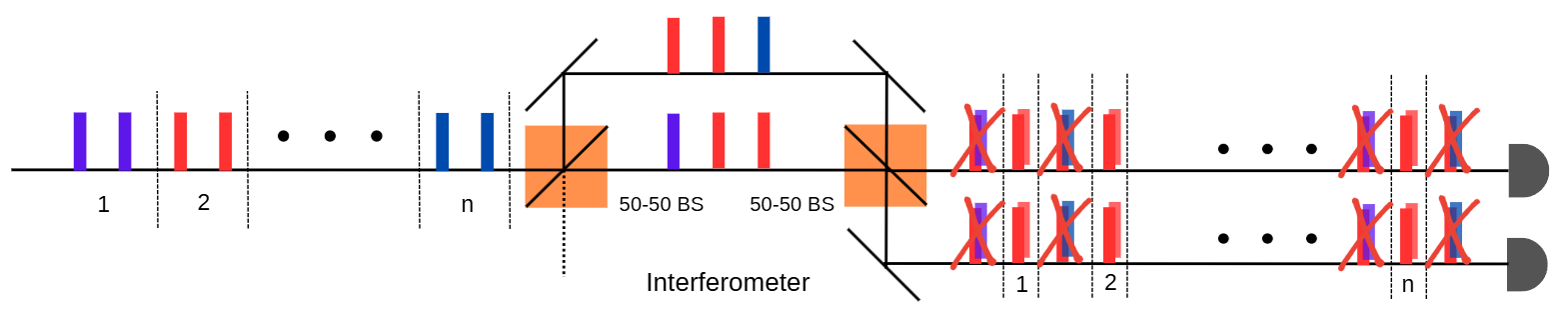}
    \caption{Mach-Zehnder interferometer with inter-round detection events ignored in post-processing. The retained ``middle'' detections interfere $e_j$ and $l_j$ within the same protocol round. The crossed-out ``outer'' detections involve interference across round boundaries, for instance between $l_j$ and $e_{j+1}$, and are coarse-grained over. After this coarse-graining, each round has $2^2=4$ retained click patterns.}
    \label{fig:IgnoreClicks}
\end{figure}

As shown in \Cref{fig:IgnoreClicks}, the detectors corresponding to the middle time-bins are kept, which correspond to the intra-round interferences, while those for the outer time-bins, which correspond to the inter-round interferences, are coarse-grained over. Each of the kept output modes thus depend on the input modes for a single round. Thus, the setup satisfies the conditions of \cref{cor:passive-linear-optics} so that we can conclude that the final POVM is product. Moreover, if the threshold detectors have no time dependent characteristics, then the POVM is IID.

\subsection{Three-state protocol detection setup}
\label{subsec:cow}

A very similar argument applies to the detection setup of the time-bin encoded three-state protocol \cite{boaron_secure_2018} --- sometimes referred to as an IID variant of the coherent one way (COW) protocol \cite{stucki_fast_2005}.

The detection setup, depicted in \cref{fig:SetupWeUse}, uses a $t/(1-t)$ beam splitter to make a passive basis choice. One output of the beam splitter leads to a threshold detector that measures the signal's time of arrival; the other, the monitoring line, uses a Mach-Zehnder interferometer to measure interference between pulses.

\begin{figure}[!ht]
    \centering
    \includegraphics[width=\linewidth]{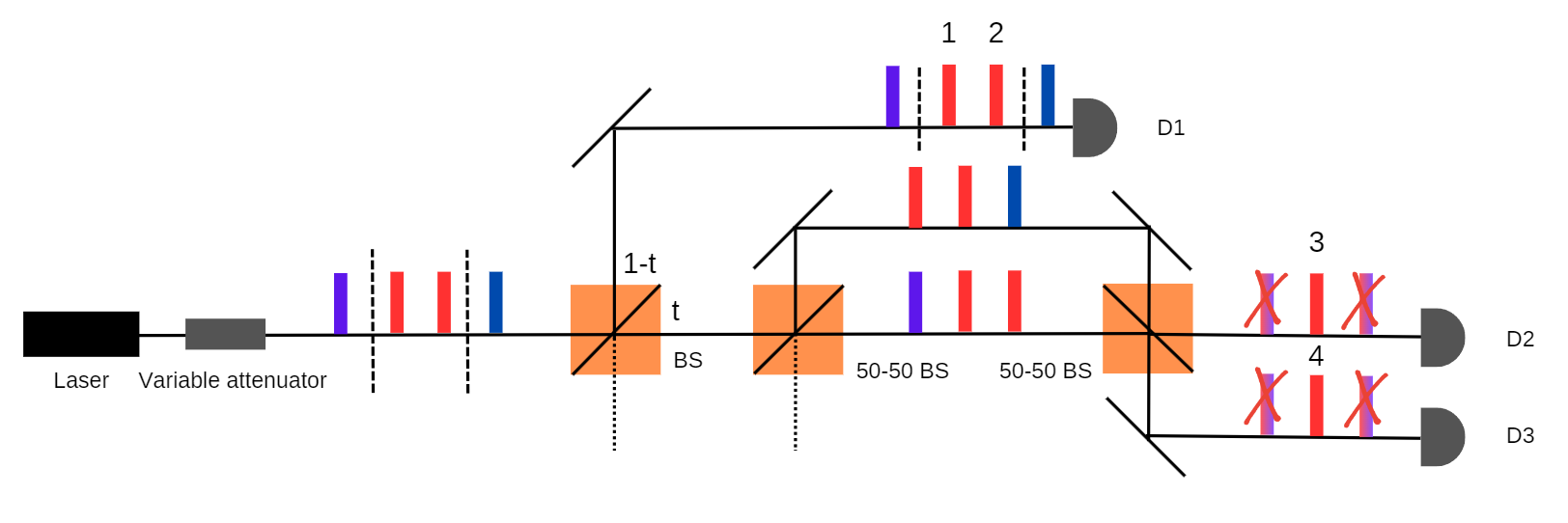}
    \caption{Optical setup of the three-state/COW-style receiver. Modes 1 and 2 measure time of arrival and are retained round-local modes. Modes 3 and 4 are retained monitoring-line modes whose amplitudes depend only on the two time-bins in the same protocol round. The unlabelled monitoring modes correspond to inter-round interference and are coarse-grained over.}
    \label{fig:SetupWeUse}
\end{figure}

Modes 1 and 2 in \cref{fig:SetupWeUse} measure time of arrival within the round. Modes 3 and 4 are retained monitoring-line modes whose output modes depend only on the two time-bin modes in the same round. The unlabelled monitoring modes correspond to inter-round interferences and are coarse-grained over. The kept part of the linear-optical mode transformation is thus round-local, and \cref{cor:passive-linear-optics} therefore implies that the resulting coarse-grained POVM has a tensor-product structure.  The same reasoning applies to higher-dimensional variants with more time-bins such as in Ref.~\cite{sulimany_highcow_2025}.

If the basis choice is implemented actively rather than by the passive beam splitter shown in \Cref{fig:SetupWeUse}, the same conclusion follows from \Cref{cor:active-basis-choice}, provided the active setting is sampled independently in each round.

\subsubsection{QKD key rate analysis}

After enforcing product structure --- and IID structure when the same single-round POVM is used in every round --- one can apply standard security-proof techniques. For example, Ref.~\cite{lavie_improved_2022} uses the postselection technique \cite{christandl_post-selection_2009, nahar_postselection_2024} to reduce the analysis to proving security against IID attacks, and then numerically computes the error rate via a semidefinite program (SDP) \cite[Section III]{primaatmaja2019versatile}.

Our work provides a ``free'' improvement to this structural reduction by removing the requirement of an extra vacuum pulse in Ref.~\cite{lavie_improved_2022}: the product and IID structure can be recovered purely by classical post-processing. An alternative that still uses the postselection technique would be to use all the single-round data (as opposed to just the phase error rate) for the IID key rates \cite{winick_reliable_2018,George_numfinite_2021,tupkary_security_2024,nahar_postselection_2024}. This would give better key rates at the cost of more computation. Entropy-accumulation-based techniques \cite{arqand_marginal-constrained_2025,tupkary2026rigorous} can also be straightforwardly applied to obtain even better finite-size performance.

In this work we focus on the product-structure reduction needed to apply standard security-proof approaches to time-bin encoded QKD protocols, and we therefore do not compute key rates explicitly.

\section{Conclusion}
\label{sec:conclusion}

In this work, we identified structural conditions under which coarse-graining of detector outcomes recovers a tensor-product POVM for time-bin encoded QKD protocols: the detector channel on the kept detectors must be round-local, and the coarse-grained detector POVM (on the output space) must contain no correlations across kept rounds.

Physically, for linear optical detection setups the first condition on the detector channel corresponds to the mode transformation for the detector assigned to round $j$ not involving any input modes from any round $i\neq j$. The second condition on the output detector POVM can be physically interpreted as the detectors having dead times and afterpulsing effects that do not affect detectors in following rounds. Note that this condition does not require detectors to have no dead times or afterpulsing --- just that the correlation time is short enough that it is entirely contained in a single round. In terms of protocol design, this can be accomplished by having a larger time delay between rounds than within rounds. Alternately, for the purposes of QKD security proofs, this can also be accomplished by discarding rounds following detection events \cite{nahar_phd_2025, tupkary2025phase,wang2025phase}.

Applied to common time-bin encoded receivers such as the Mach-Zehnder interferometer and the three-state detection setup, our framework removes the need for the additional vacuum pulse introduced in previous analyses \cite{sulimany_highcow_2025,lavie_improved_2022} solely to recover product structure. Consequently, no assumption is required that Eve forwards the vacuum pulse unchanged.
frameworks.

\begin{acknowledgments}
We thank Devashish Tupkary for valuable feedback on drafts of this paper.
Most of the work was performed at the Institute for Quantum Computing, at the University of Waterloo, which is supported by Innovation, Science, and Economic Development Canada. The research has been supported by NSERC under the Discovery Grants Program, Grant No. 341495. This research has been supported by Alliance QUINT.
\end{acknowledgments}

\bibliographystyle{unsrt}
\bibliography{referencesTemp}

@article{tupkary_security_2024,
	title = {Security proof for variable-length quantum key distribution},
	volume = {6},
	url = {https://link.aps.org/doi/10.1103/PhysRevResearch.6.023002},
	doi = {10.1103/PhysRevResearch.6.023002},
	number = {2},
	urldate = {2024-04-30},
	journal = {Physical Review Research},
	author = {Tupkary, Devashish and Tan, Ernest Y.-Z. and Lütkenhaus, Norbert},
	month = apr,
	year = {2024},
	note = {Publisher: American Physical Society},
	pages = {023002},
}

@misc{nahar_postselection_2024,
	title = {Postselection technique for optical {Quantum} {Key} {Distribution} with improved de {Finetti} reductions},
	url = {http://arxiv.org/abs/2403.11851},
	doi = {10.48550/arXiv.2403.11851},
	urldate = {2024-04-02},
	publisher = {arXiv},
	author = {Nahar, Shlok and Tupkary, Devashish and Zhao, Yuming and Lütkenhaus, Norbert and Tan, Ernest},
	month = mar,
	year = {2024},
	note = {arXiv:2403.11851 [math-ph, physics:physics, physics:quant-ph]},
}

@article{winick_reliable_2018,
	title = {Reliable numerical key rates for quantum key distribution},
	volume = {2},
	issn = {2521-327X},
	url = {http://arxiv.org/abs/1710.05511},
	doi = {10.22331/q-2018-07-26-77},
	language = {en},
	urldate = {2023-11-15},
	journal = {Quantum},
	author = {Winick, Adam and Lütkenhaus, Norbert and Coles, Patrick J.},
	month = jul,
	year = {2018},
	note = {arXiv:1710.05511 [quant-ph]},
	pages = {77},
}

@article{inoue_differential_2002,
	title = {Differential {Phase} {Shift} {Quantum} {Key} {Distribution}},
	volume = {89},
	url = {https://link.aps.org/doi/10.1103/PhysRevLett.89.037902},
	doi = {10.1103/PhysRevLett.89.037902},
	number = {3},
	urldate = {2024-01-20},
	journal = {Physical Review Letters},
	author = {Inoue, Kyo and Waks, Edo and Yamamoto, Yoshihisa},
	month = jun,
	year = {2002},
	note = {Publisher: American Physical Society},
	pages = {037902},
}

@article{stucki_fast_2005,
	title = {Fast and simple one-way quantum key distribution},
	volume = {87},
	issn = {0003-6951},
	url = {https://doi.org/10.1063/1.2126792},
	doi = {10.1063/1.2126792},
	number = {19},
	urldate = {2024-01-20},
	journal = {Applied Physics Letters},
	author = {Stucki, Damien and Brunner, Nicolas and Gisin, Nicolas and Scarani, Valerio and Zbinden, Hugo},
	month = nov,
	year = {2005},
	pages = {194108},
}

@article{boaron_secure_2018,
	title = {Secure {Quantum} {Key} {Distribution} over 421 km of {Optical} {Fiber}},
	volume = {121},
	url = {https://link.aps.org/doi/10.1103/PhysRevLett.121.190502},
	doi = {10.1103/PhysRevLett.121.190502},
	number = {19},
	urldate = {2024-01-20},
	journal = {Physical Review Letters},
	author = {Boaron, Alberto and Boso, Gianluca and Rusca, Davide and Vulliez, Cédric and Autebert, Claire and Caloz, Misael and Perrenoud, Matthieu and Gras, Gaëtan and Bussières, Félix and Li, Ming-Jun and Nolan, Daniel and Martin, Anthony and Zbinden, Hugo},
	month = nov,
	year = {2018},
	note = {Publisher: American Physical Society},
	pages = {190502},
}

@article{lavie_improved_2022,
	title = {Improved {Coherent} {One}-{Way} {Quantum} key {Distribution} for {High}-{Loss} {Channels}},
	volume = {18},
	issn = {2331-7019},
	url = {https://link.aps.org/doi/10.1103/PhysRevApplied.18.064053},
	doi = {10.1103/PhysRevApplied.18.064053},
	language = {en},
	number = {6},
	urldate = {2023-06-24},
	journal = {Physical Review Applied},
	author = {Lavie, Emilien and Lim, Charles C.-W.},
	month = dec,
	year = {2022},
	pages = {064053},
}

@article{christandl_post-selection_2009,
	title = {Post-selection technique for quantum channels with applications to quantum cryptography},
	volume = {102},
	issn = {0031-9007, 1079-7114},
	url = {http://arxiv.org/abs/0809.3019},
	doi = {10.1103/PhysRevLett.102.020504},
	number = {2},
	urldate = {2023-03-20},
	journal = {Physical Review Letters},
	author = {Christandl, Matthias and Koenig, Robert and Renner, Renato},
	month = jan,
	year = {2009},
	note = {arXiv:0809.3019 [quant-ph]},
	pages = {020504},
}

@article{sulimany_highcow_2025,
author={Sulimany, Kfir
and Pelc, Guy
and Dudkiewicz, Rom
and Korenblit, Simcha
and Eisenberg, Hagai S.
and Bromberg, Yaron
and Ben-Or, Michael},
title={High-dimensional coherent one-way quantum key distribution},
journal={npj Quantum Information},
year={2025},
month={Jan},
day={29},
volume={11},
number={1},
pages={16},
issn={2056-6387},
doi={10.1038/s41534-025-00965-7},
url={https://doi.org/10.1038/s41534-025-00965-7}
}

@article{George_numfinite_2021,
  title = {Numerical calculations of the finite key rate for general quantum key distribution protocols},
  author = {George, Ian and Lin, Jie and L\"utkenhaus, Norbert},
  journal = {Phys. Rev. Res.},
  volume = {3},
  issue = {1},
  pages = {013274},
  numpages = {25},
  year = {2021},
  month = {Mar},
  publisher = {American Physical Society},
  doi = {10.1103/PhysRevResearch.3.013274},
  url = {https://link.aps.org/doi/10.1103/PhysRevResearch.3.013274}
}

@article{sandfuchs2025security,
  title={Security of differential phase shift QKD from relativistic principles},
  author={Sandfuchs, Martin and Haberland, Marcus and Vilasini, Venkatesh and Wolf, Ramona},
  journal={Quantum},
  volume={9},
  pages={1611},
  year={2025},
  publisher={Verein zur F{\"o}rderung des Open Access Publizierens in den Quantenwissenschaften},
  url={https://quantum-journal.org/papers/q-2025-01-27-1611/}
}

@article{mosca2018cybersecurity,
  title={Cybersecurity in an era with quantum computers: will we be ready?},
  author={Mosca, Michele},
  journal={IEEE Security \& Privacy},
  volume={16},
  number={5},
  pages={38--41},
  year={2018},
  publisher={IEEE}
}

@misc{arqand_marginal-constrained_2025,
    title = {Marginal-constrained entropy accumulation theorem},
    url = {http://arxiv.org/abs/2502.02563},
    doi = {10.48550/arXiv.2502.02563},
    urldate = {2025-07-09},
    publisher = {arXiv},
    author = {Arqand, Amir and Tan, Ernest Y.-Z.},
    month = apr,
    year = {2025},
    note = {arXiv:2502.02563 [quant-ph]},
}

@article{dupuis_entropy_2020,
    title = {Entropy {Accumulation}},
    volume = {379},
    issn = {1432-0916},
    url = {https://doi.org/10.1007/s00220-020-03839-5},
    doi = {10.1007/s00220-020-03839-5},
    language = {en},
    number = {3},
    urldate = {2024-07-19},
    journal = {Communications in Mathematical Physics},
    author = {Dupuis, Frédéric and Fawzi, Omar and Renner, Renato},
    month = nov,
    year = {2020},
    pages = {867--913},
}

@article{metger_security_2023,
    title = {Security of quantum key distribution from generalised entropy accumulation},
    volume = {14},
    copyright = {2023 The Author(s)},
    issn = {2041-1723},
    url = {https://www.nature.com/articles/s41467-023-40920-8},
    doi = {10.1038/s41467-023-40920-8},
    language = {en},
    number = {1},
    urldate = {2024-10-05},
    journal = {Nature Communications},
    author = {Metger, Tony and Renner, Renato},
    month = aug,
    year = {2023},
    note = {Publisher: Nature Publishing Group},
    pages = {5272},
}

@article{primaatmaja2019versatile,
  title={Versatile security analysis of measurement-device-independent quantum key distribution},
  author={Primaatmaja, Ignatius William and Lavie, Emilien and Goh, Koon Tong and Wang, Chao and Lim, Charles Ci Wen},
  journal={Physical Review A},
  volume={99},
  number={6},
  pages={062332},
  year={2019},
  publisher={APS},
  url={https://journals.aps.org/pra/abstract/10.1103/PhysRevA.99.062332}
}

@article{wang2025phase,
  title={Phase error estimation for passive detection setups with imperfections and memory effects},
  author={Wang, Zhiyao and Tupkary, Devashish and Nahar, Shlok},
  journal={arXiv preprint arXiv:2508.21486},
  year={2025},
  url={https://arxiv.org/abs/2508.21486}
}

@article{tupkary2025phase,
  title={Phase error rate estimation in QKD with imperfect detectors},
  author={Tupkary, Devashish and Nahar, Shlok and Sinha, Pulkit and L{\"u}tkenhaus, Norbert},
  journal={Quantum},
  volume={9},
  pages={1937},
  year={2025},
  publisher={Verein zur F{\"o}rderung des Open Access Publizierens in den Quantenwissenschaften},
  url={https://quantum-journal.org/papers/q-2025-12-11-1937/}
}

@phdthesis{nahar_phd_2025,
  author       = {Nahar, Shlok},
  title        = {A proof-technique-independent framework for detector imperfections in QKD},
  school       = {University of Waterloo},
  year         = {2026},
  url = {https://uwspace.uwaterloo.ca/items/ff6bda01-eefc-4bd7-bc1c-6207afd36142}
}

@article{koashi_simple_2005,
  title={Simple security proof of quantum key distribution via uncertainty principle},
  author={Koashi, Masato},
  journal={arXiv preprint quant-ph/0505108},
  year={2005},
  url={https://arxiv.org/abs/quant-ph/0505108}
}

@article{koashi_simple_2009,
  title = {Simple Security Proof of Quantum Key Distribution Based on Complementarity},
  author = {Koashi, M.},
  year = {2009},
  month = apr,
  journal = {New Journal of Physics},
  volume = {11},
  number = {4},
  pages = {045018},
  issn = {1367-2630},
  doi = {10.1088/1367-2630/11/4/045018},
  urldate = {2024-11-25},
  langid = {english},
  language = {en}
}

@article{tupkary2025qkd,
  title={QKD security proofs for decoy-state {BB84}: protocol variations, proof techniques, gaps and limitations},
  author={Tupkary, Devashish and Tan, Ernest Y-Z and Nahar, Shlok and Kamin, Lars and L{\"u}tkenhaus, Norbert},
  journal={arXiv preprint arXiv:2502.10340},
  year={2025},
  url={https://arxiv.org/abs/2502.10340}
}

@article{tupkary2026rigorous,
  title={A rigorous and complete security proof of decoy-state BB84 quantum key distribution},
  author={Tupkary, Devashish and Nahar, Shlok and Arqand, Amir and Tan, Ernest Y-Z and L{\"u}tkenhaus, Norbert},
  journal={arXiv preprint arXiv:2601.18035},
  year={2026},
  url={https://arxiv.org/abs/2601.18035}
}

@article{tomamichel_uncertainty_2011,
  title = {Uncertainty {{Relation}} for {{Smooth Entropies}}},
  author = {Tomamichel, Marco and Renner, Renato},
  year = {2011},
  month = mar,
  journal = {Physical Review Letters},
  volume = {106},
  number = {11},
  pages = {110506},
  issn = {0031-9007, 1079-7114},
  doi = {10.1103/PhysRevLett.106.110506},
  urldate = {2023-09-25},
  langid = {english},
  language = {en}
}

\end{document}